\documentclass[11pt]{article}

\usepackage{fullpage}

\usepackage{amsthm,amsmath}  
\usepackage{xspace,enumerate}
\usepackage[utf8]{inputenc}
\usepackage{thmtools}
\usepackage{thm-restate}
\usepackage{authblk}
\usepackage[hypertexnames=false,colorlinks=true,urlcolor=Blue,citecolor=Green,linkcolor=BrickRed]{hyperref}
\usepackage[dvipsnames,usenames]{color}
\usepackage{graphicx}
\usepackage[noadjust]{cite}
\usepackage{microtype}
\usepackage{todonotes}
\usepackage{thmtools}
\usepackage{thm-restate}
\usepackage[capitalise]{cleveref}
\usepackage{algorithmic}
\usepackage[pagewise]{lineno}
\nolinenumbers

\theoremstyle{plain}
\newtheorem{theorem}{Theorem}
\newtheorem{lemma}[theorem]{Lemma}

\usepackage{todonotes}

\newcommand{\cut}{{$cut$}}
\newcommand{\da}{\downarrow}

\author[1]{Pawe\l{} Gawrychowski}
\author[2]{Shay Mozes}
\author[3]{Oren Weimann}

\affil[1]{
University of Wroc\l{}aw, Poland\\
\href{mailto:gawry@cs.uni.wroc.pl}{gawry@cs.uni.wroc.pl}}

\affil[2]{
The Interdisciplinary Center Herzliya, Israel\\
\href{mailto:smozes@idc.ac.il}{smozes@idc.ac.il}}

\affil[3]{
University of Haifa, Israel\\
\href{mailto:oren@cs.haifa.ac.il}{oren@cs.haifa.ac.il}
}

\date{}

\title{A Note on a Recent Algorithm for Minimum Cut}

\begin{document}

\maketitle

\thispagestyle{empty}

\begin{abstract}

Given an undirected edge-weighted graph $G=(V,E)$ with $m$ edges and $n$ vertices, the minimum cut problem asks to find a subset of vertices $S$ such that the total weight of all edges between $S$ and $V \setminus S$ is minimized. Karger's longstanding $O(m \log^3 n)$ time randomized algorithm for this problem was very recently improved in two independent works to $O(m \log^2 n)$ [ICALP'20] and to $O(m \log^2 n + n\log^5 n)$ [STOC'20]. These two algorithms use different approaches and techniques. In particular, while the former is faster, the latter has the advantage that it can be used to obtain efficient algorithms in the cut-query and in the streaming models of computation. 
In this paper, we show how to simplify and improve the algorithm of [STOC'20] to $O(m \log^2 n  + n\log^3 n)$. We obtain this by replacing a randomized algorithm that, given a spanning tree $T$ of $G$, finds in $O(m \log n+n\log^4 n)$ time a minimum cut of $G$ that 2-respects (cuts two edges of) $T$ with a simple $O(m \log n+n\log^2 n)$ time deterministic algorithm for the same problem.
 
\end{abstract}

\section{Introduction}

In his seminal work in 1996,  Karger~\cite{Karger} showed how to find the minimum cut of an  edge-weighted undirected graph in $O(m\log^3 n)$ time. The first step of his algorithm is a procedure that, given an undirected edge-weighted graph $G$, produces in $O(m + n \log^3 n)$ time a collection of $O(\log n)$ spanning trees of $G$ such that w.h.p the minimum cut {\em 1-} or {\em 2-respects} some tree in the collection.
That is, one of the trees is such that at most two of its edges cross the minimum cut (these edges are said to {\em determine} the cut). The minimum cut is then found by examining each tree $T$ of the $O(\log n)$ trees and finding the minimum cut that 1- or 2-respects $T$.  
Since the minimum cut that 1-respects $T$ can be easily found in $O(m)$ time~\cite[Lemma 5.1]{Karger}, the main challenge is to find the minimum cut that 2-respects $T$. 

Karger showed that the minimum cut that 2-respects a given tree can be found in $O(m \log^2 n)$ time. 
This was very recently improved in two independent works: In~\cite{ourICALP2020} we obtained an $O(m \log n)$ deterministic algorithm,\footnote{We also showed in~\cite{ourICALP2020} that the first step of producing the collection of spanning trees can be performed (using a randomized algorithm) in $O(m \log^2 n)$ time, leading to an $O(m \log^2 n)$ time randomized algorithm for min cut.}  and in~\cite{Danupon} Mukhopadhyay and Nanongkai obtained an $O(m \log n + n\log^4 n)$ randomized algorithm. 
These two results use different techniques.
Even though~\cite{ourICALP2020} dominates~\cite{Danupon} for the entire range of graphs densities, the approach of~\cite{Danupon} has two notable advantages: (1) it can be extended to other models, namely to find the minimum cut using $\tilde O(n)$ cut queries, or using $\tilde O(n)$ space and $O(\log n)$ passes in a streaming algorithm, and (2) the approach of~\cite{Danupon} can be seen as a reduction from min cut to geometric two-dimensional orthogonal range counting/sampling/reporting data structures~\cite{Chaz88}. Therefore, special cases or future improvements of such data structures will imply improvements to min cut. For example, for the special case of {\em unweighted} undirected graphs, \cite{Danupon} use an improved data structure for orthogonal range counting~\cite{CP10} and range rank/select~\cite{BGKS15}
(that are then used to design an improved data structure for orthogonal range reporting and, finally, orthogonal range sampling).
This yields an $O(m \sqrt{\log n}+n\log^4 n)$ time randomized algorithm for finding a 2-respecting min cut in unweighted graphs, 
and hence an $O(m \log^{3/2} n+n\log^5 n)$ time algorithm for min cut in such graphs.

\paragraph{Our results.}
In this paper we show how to simplify the algorithm of Mukhopadhyay and Nanongkai, improve its running time to $O(m \log n + n\log^2 n)$, and turn it deterministic. By Karger's reasoning, this then implies a randomized min cut algorithm working in $O(m\log^{2}n+n\log^{3}n)$ time.
In fact, one can also apply his $\log\log n$ speedup that exploits the fact that 1-respecting cuts can be found faster
than 2-respecting cuts. 
As explained in~\cite[Section 4]{usarxiv}, by appropriately tweaking the parameters we
can obtain a randomized min cut algorithm working in $O(m\log^{2}n/\log\log n+n\log^{3+\epsilon}n)$ time.
Interestingly, with our improvement, the reduction to the geometric data structure is now a clean black box reduction to just orthogonal range counting (no sampling/reporting is required). This allows us to obtain the following new results:
(1) an $O(m\log n + n^{1+\epsilon})$-time randomized algorithm for min cut in weighted graphs for any fixed $\epsilon >0$ (this dominates all previous results for $m = \Omega(n^{1+\epsilon})$), and (2) an $O(m \log^{3/2} n+n\log^3 n)$ time randomized algorithm for min cut in unweighted graphs.

\paragraph{Roadmap.} In Section~\ref{sec:prelims} we describe the algorithm of Mukhopadhyay and Nanongkai~\cite{Danupon}. Our description uses slightly different terminology than~\cite{Danupon}, but all the ideas described in Section~\ref{sec:prelims} are taken from~\cite{Danupon}. In Section~\ref{sec:Simplification} we describe our simplification of~\cite{Danupon} and in Section~\ref{sec:Improvements} we show how to use it to achieve faster algorithms for unweighted graphs and for dense graphs. 

\section{The Algorithm of Mukhopadhyay and Nanongkai~\cite{Danupon}}
\label{sec:prelims}
In this section we describe the algorithm of Mukhopadhyay and Nanongkai~\cite{Danupon} for finding the pair of edges $\{e,e'\}$ determining the minimum cut (observe that the cut determined by $\{e,e'\}$ is unique and consists of all edges $(u,v) \in G$ such that the $u$-to-$v$ path in $T$ contains exactly one of $\{e,e'\}$). 

The algorithm begins by partitioning the tree $T$ into a set $\mathcal{P}$ of edge-disjoint paths (called {\em heavy paths}~\cite{LinkCutTree}) such that any root-to-leaf path in $T$ intersects at most $\log n$ heavy paths. 

\subsection{Two edges in the same heavy path}\label{subsection:same}
Consider first the case where the minimum cut is determined by two edges of the same path $P\in \mathcal{P}$. Finding these two edges then boils down to finding the smallest element in the $(\ell-1)\times (\ell-1)$ matrix $M$ where $\ell$ is the length of $P$ and $M[i,j]$ is the weight of the cut determined by the $i$'th and the $j$'th edges of $P$. An important contribution of Mukhopadhyay and Nanongkai is in observing that the matrix $M$ is a \emph{Partial Monge} matrix. That is, for any $i\neq j$, it holds that $M[i,j] - M[i,j+1] \ge M[i+1,j] - M[i+1,j+1]$.\footnote{Mukhopadhyay and Nanongkai reversed the order of rows so in their presentation the condition was $M[i,j] - M[i,j+1] \le M[i+1,j] - M[i+1,j+1]$.} They then describe an algorithm that finds the smallest element in $M$ by inspecting only $O(\ell \cdot  \log^2 \ell)$ entries of $M$. Instead, one could use the faster algorithm by Klawe and Kleitman~\cite{KK89} that requires only $O(\ell\cdot \alpha(\ell))$ inspections (where $\alpha$ is the inverse-Ackermann function).   Lemma~\ref{lemma:Chazelle} below shows that each inspection can be done in $O(\log n)$ time.
Thus, in $O(\ell \cdot  \alpha(\ell) \cdot \log n)$ time one can find the minimum cut determined by two edges of $P$. Since paths in $\mathcal{P}$ are disjoint, doing this for all paths in $\mathcal{P}$ this takes overall $O(n \cdot  \alpha(n) \cdot \log n)$ time. 

\subsection{Two edges in different heavy paths}\label{subsection:different}
Now consider the case where the minimum cut is determined by two edges belonging to different paths in $\mathcal{P}$. Another significant insight of Mukhopadhyay and Nanongkai is that  there is no  need to check every pair of paths $P,Q \in \mathcal{P}$ but only a small subset of {\em interesting} path pairs as explained next. 
Let \cut($e$,$e'$) denote the weight of the cut determined by edges $\{e,e'\}$.
Let $T_{e}$ denote the subtree of $T$ rooted at the lower (i.e., further from the root) endpoint of $e$.
If $e' \in T_e$ then we say that $e'$ is a {\em descendant} of $e$. If $e'$ is not a descendant of $e$ and $e$ is not a descendant of $e'$ then we say that $e$ and $e'$ are {\em independent}. 

\paragraph{Cross-interested edges.} An edge $e \in T$ is said to be {\em cross-interested} in  an edge $e' \in T \setminus T_e$ if $$w(T_e)<2w(T_{e},T_{e'})$$ where $w(T_{e})$ is the total weight of edges between $T_{e}$ and $V(G) \setminus V(T_{e})$ and $w(T_{e},T_{e'})$ is the total weight of edges between $T_{e}$ and $T_{e'}$. That is, $e$ is cross-interested in $e'$ if more than half the edge weight going out of $T_e$ goes into $T_{e'}$. 
Observe that if the minimum cut is determined by independent edges $\{e,e'\}$ then $e$ must be cross-interested in $e'$ (and vice versa) because otherwise \cut($e$,$e'$) = $w(T_{e}) + w(T_{e'}) - 2w(T_{e},T_{e'}) >  w(T_{e'})$ (i.e. the cut determined by the single edge $e'$ has smaller weight, a contradiction).
This means that there is no need to check every pair of independent edges, only ones that are cross-interested. 
It is easy to see that for any tree-edge $e$, all the edges that $e$ is cross-interested in form a single path $C_e$ in $T$ going down from the root to some node $c_e$.

\paragraph{Down-interested edges.}  
An edge $e \in T$ is said to be {\em down-interested} in  an edge $e' \in T_e$ if $$w(T_{e})<2w(T_{e'},T\setminus T_{e})$$ where $w(T_{e'},T\setminus T_{e})$ is the total weight of edges between $T_{e'}$ and $V(G)\setminus V(T_{e})$. That is, $e$ is down-interested in $e'$ if more than half the edge weight going out of $T_e$ originates in $T_{e'}$.
Observe that if the minimum cut is determined by edges $e$ and $e'$ where $e'$ is a descendant of $e$, then $e$ must be down-interested in $e'$ because otherwise \cut($e$,$e'$) = $w(T_{e}) + w(T_{e'}) - 2w(T_{e'},T \setminus T_e) >  w(T_{e'})$ (again, a contradiction). 
For convenience, define that $e$ is down-interested in all of its ancestor edges.
This means that we only need to check pairs of descendant edges that are down-interested in each other. Furthermore, for any tree-edge $e$, all the edges that $e$ is down-interested in form a single path $D_e$ in $T$ going down from the root to some node $d_e$.

A third important realization of Mukhopadhyay and Nanongkai is that a geometric range searching data structure of Chazelle~\cite{Chaz88} can be used to efficiently determine whether an edge $e$ is interested in an edge $e'$. This is described in the following lemma.
\begin{lemma}\label{lemma:Chazelle} 
Given a graph $G$ and a spanning tree $T$, we can construct in $O(m\log n)$ time a data structure that, given any two edges $e,e'$, can report in $O(\log n)$ time (1) the value  \cut($e$,$e'$), (2) whether $e$ is cross-interested in $e'$, and (3) whether $e$ is down-interested in $e'$.
\end{lemma}
\begin{proof}
In $O(n)$ time we construct a data structure that can answer lowest common ancestor queries on $T$ in constant time~\cite{HT1984}. 
For every node $v \in T$, let $\hat{v}\in [n]$ denote the visiting time of $v$ in a postorder traversal of $T$ and let $\hat{v}^\da$ denote the minimum visiting time of a node in the subtree of $T$ rooted at $v$. 
Let $w(v^\da)$ be the total weight of edges with exactly one endpoint in the subtree of $T$ rooted at $v$.
As also done by Karger~\cite{Karger}, in a bottom up fashion (in linear time) we compute $\hat{v}$, $\hat{v}^\da$, and $w(v^\da)$ for every $v \in T$. 
We map each edge $(u,v) \in G$ to the  point $(\hat{u},\hat{v})$ in the two-dimensional plane. On this set of $m$ points we construct Chazelle's {\em 2D orthogonal range searching} data structure~\cite{Chaz88}. This data structure is constructed in $O(m\log n)$ time and can report in $O(\log n)$  time the total weight of all points in any given axis-aligned rectangle. 

Consider any two edges $e$ and $e'$. Let $u$ and $v$ be the lower endpoints of $e$ and $e'$, respectively. Note that $w(T_e) = w(u^\da)$ and $w(T_{e'}) = w(v^\da)$.
Consider first the case that $e$ and $e'$ are independent. Deciding whether $e$ is cross-interested in $e'$ reduces to computing $w(T_e, T_{e'})$ which is obtained by a range query to the rectangle $[\hat{u}^\da,\hat{u}]\times[\hat{v}^\da,\hat{v}]$. The value \cut($e$,$e'$) is computed as $w(v^\da)+ w(u^\da)-2w(T_e,T_{e'})$. 

Now consider the case that $e'$ is a descendant of $e$.  Then deciding whether $e$ is down-interested in $e'$ reduces to computing $w(T_{e'},T \setminus T_e)$ which is obtained as the sum of the answers to the rectangles 
$[\hat{v}^\da,\hat{v}]\times[1,\hat{u}-1]$ and $[\hat{v}^\da,\hat{v}]\times[\hat{u}+1,n]$. The value \cut($e$,$e'$) is computed as $w(u^\da)+ w(v^\da)-2w(T_{e'},T\setminus T_e)$. 

Finally, if $e$ is a descendant of $e'$, then we always report that $e$ is down-interested in $e'$. The value \cut($e$,$e'$) is computed (symmetrically to the above) as $w(v^\da)+ w(u^\da)-2w(T_e,T\setminus T_{e'})$.  
\end{proof}

\paragraph{Interesting path pairs.}
Recall that the goal is two find the two tree-edges $\{e,e'\}$ that determine the minimum cut and we know that these edges belong to different heavy paths $P,Q \in \mathcal{P}$.  
A tree-edge $e$ is said to be {\em interested in a path $P$} in $\mathcal{P}$ if it is cross-interested or down-interested in some edge of $P$. Notice that by the above, any tree-edge $e$ is interested in only $O(\log n)$ paths.
Define a pair of paths $P,Q \in \mathcal{P}$ to be an {\em interesting pair} if $P$ has an edge interested in $Q$ and $Q$ has an edge interested in $P$. 

Notice that the number of interesting pairs $P,Q$ is only $O(n\log n)$.
However, Mukhopadhyay and Nanongkai do not identify all the interesting pairs. Instead, they
apply a complicated random sampling scheme in order to find the best pair with high probability. This sampling makes their algorithm randomized and its running time $O(m \log n+n\log^4 n)$. 
In Section~\ref{subsection:Finding} we show how to replace the random sampling step with a much simpler deterministic algorithm that finds {\em all} interesting pairs of paths. Our algorithm is also faster, taking $O(m \log n+n\log^2 n)$ time. Then, for each interesting pair $P,Q$, we (conceptually) contract all tree-edges except those in $P$ that are interested in $Q$ and those in $Q$ that are interested in $P$, and run the solution from Section~\ref{subsection:same} on the resulting paths. This last step is very similar to the corresponding step in~\cite{Danupon}. We explain it in detail in Section~\ref{subsection:Checking}.   

\section{The Simplification}\label{sec:Simplification}

For every edge $e\in T$, let $C_e$ ($D_e$) denote the path in $T$ consisting of all the edges that $e$ is  cross-interested (down-interested) in. The path $C_e$ ($D_e$) starts at the root and terminates at some node denoted $c_e$ ($d_e$). For every $e\in T$, we compute $c_e$ and $d_e$ in  $O(\log^2 n)$ time. In contrast to~\cite{Danupon}, we do this deterministically by using a {\em centroid decomposition}.
     
\subsection{Finding interesting path pairs}\label{subsection:Finding}
A node $v\in T$ is a centroid if every connected component of $T\setminus\{v\}$ consists of at most ${|T|}/{2}$
nodes. The centroid decomposition of $T$ is defined recursively by first choosing a centroid $v\in T$
and then recursing on every connected component of $T\setminus\{v\}$.
We assume $T$ is a binary tree (we can replace a node of degree $d$ with a binary tree of size $O(d)$ where internal edges have weight $\infty$ and edges incident to leaves have their original weight). We also assume we have a centroid decomposition of $T$ (we can compute a centroid decomposition of every tree in $O(n\log n)$ time so overall in $O(n\log^2 n)$ time). 
To compute $c_e$, consider the (at most) three edges $e_1,e_2,e_3$ incident to the centroid node. Using Lemma~\ref{lemma:Chazelle}, we check in $O(\log n)$ time whether $e$ is cross-interested in $e_1$, in $e_2$, and in $e_3$. From this we can deduce in which connected component $c_e$ lies, and we continue recursively there. Since the recursion depth is $O(\log n)$, we find $c_e$ after $O(\log^2 n)$ time so overall we spend $O(n \log^2 n)$ time. We compute $d_e$ similarly (querying Lemma~\ref{lemma:Chazelle} for down-interested rather than cross-interested).  

\subsection{Checking interesting path pairs}\label{subsection:Checking}

For each interesting pair of heavy paths $P,Q$, we will store a list of the edges of $Q$ that are interested in $P$ and vice versa. Recall that, since each edge is interested in $O(\log n)$ heavy paths, the number of interesting pairs of heavy paths is only $O(n\log n)$. Moreover, the total length of all the lists is also $O(n\log n)$.  

We first show how to compute a list of all interesting pairs of heavy paths. By going over all the edges of $T$ we prepare for each heavy path $P \in \mathcal{P}$ a list of all the heavy paths $Q$ s.t. an edge of $P$ is interested in an edge of $Q$. The total size of all these lists is $O(n\log n)$ and they can be computed in $O(n \log^2 n)$ time using Lemma~\ref{lemma:Chazelle}.
We then sort these lists (according to some canonical order on the heavy paths). Then, for every $P \in \mathcal{P}$ we go over all heavy paths $Q$ that $P$ is interested in. For each such $Q$ we determine in $O(\log n)$ time whether $Q$ is also interested in $P$ using binary search on the list of $Q$. Thus we construct the lists of all interesting pairs in total $O(n \log^2 n)$ time. 

For each interesting pair of paths $P,Q$, we construct a list of the edges of $Q$ that are interested in $P$ and vice versa as follows. We go over the edges $e$ of $T$. Let $Q$ be the heavy path containing $e$. 
For each heavy path $P$ that intersects $C_e$ or $D_e$, if $P,Q$ is an interesting pair, we add $e$ to the list of the pair $P,Q$. This takes $O(\log n)$ time since there are $O(\log n)$ such paths $P$, so the total time to construct all these lists is $O(n\log n)$.
Finally, we sort the edges on each list in $O(n\log^{2}n)$ total time.

For each interesting pair of paths $P,Q$, let $P'$ ($Q'$) denote the set of edges of $P$ ($Q$) that are interested in $Q$ ($P$).   
We find the minimum cut determined by pairs of edges $e,e'$ such that $e \in P'$ and $e' \in Q'$ in a single batch as follows.
We assume that either $P'$ is a descendant of $Q'$ (i.e. all edges in $P'$ are descendants of all edges in $Q'$) or that $P'$ is independent of $Q'$ (i.e. no edge in $P'$ is a descendant of an edge in $Q'$). Otherwise, if $P'$ is a descendant of one part of $Q'$ and independent of another part then we just split $Q'$ into two parts and handle each separately. We think of $P'$ as being oriented root-wards. If $P'$ is a descendant of $Q'$ then we orient $Q'$ root-wards, and if $P'$ is independent of $Q'$ then we orient $Q'$ leaf-wards.  
Let $M$ be the $|P'|\times |Q'|$ matrix where $M[i,j]$ is the weight of the  cut determined by the $i$'th edge of $P'$ and the $j$'th edge of $Q'$.
 We observe that the matrix $M$ is a Monge matrix (rather than \emph{Partial} Monge). That is, the Monge condition $M[i,j] - M[i,j+1] \ge M[i+1,j] - M[i+1,j+1]$ holds for any $i,j$ (and not only for $i\neq j$). This means that instead of the Klawe-Kleitman algorithm~\cite{KK89} we can use the SMAWK algorithm~\cite{SMAWK} that finds the maximum entry in $M$ by inspecting only a linear $O(|P'|+|Q'|)$ number of entries of $M$ (i.e. without the additional inverse-Ackermann term). Using Lemma~\ref{lemma:Chazelle} for each inspection, this takes $O((|P'|+|Q'|)\log n)$ time. Since the sum $\sum (|P'|+|Q'|)$ over all interesting pairs of paths is $O(n\log n)$, the overall time is $O(n\log^2 n)$.     

The proof that $M$ is Monge appears in~\cite[Claim 3.5]{Danupon}. We give one here for completeness.

\begin{lemma}	
$M[i,j] - M[i,j+1] \ge M[i+1,j] - M[i+1,j+1]$ for any $i,j$.
\end{lemma}
\begin{figure}[h!]
\begin{center}
\includegraphics[width=0.8\textwidth]{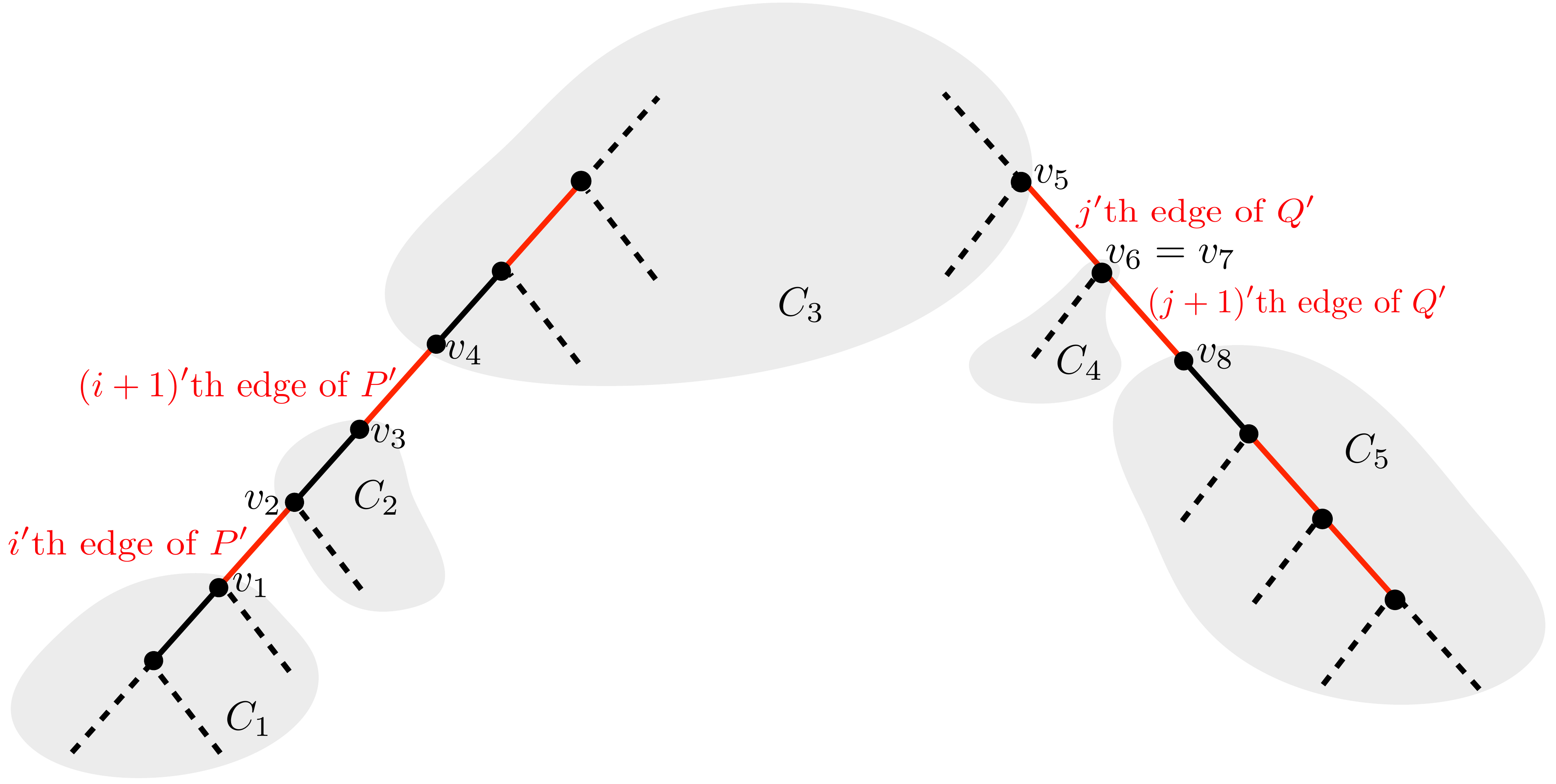}
\end{center}
\caption{Two independent heavy paths $P$ and $Q$ (solid edges). The red edges are $P'$ (oriented root-wards) and $Q'$ (oriented leaf-wards). The connected components $C_1,\ldots,C_5$ (shaded) are obtained by deleting the $i$'th and $(i+1)$'th edges of $P'$ and the $j$'th and $(j+1)$'th edges of $Q'$. \label{fig}}
\end{figure}

\begin{proof}

Recall that the order of edges in $P'$ is root-wards and that that the order of edges in $Q'$ is root-wards if $P'$ is a descendant of $Q'$ and leaf-wards if $P'$ is independent of $Q'$. 
With this order in mind, let $(v_1,v_2)$ and $(v_3,v_4)$ denote the $i$'th and $(i+1)$'th edges of $P'$ and let $(v_5,v_6)$ and $(v_7,v_8)$ denote the $j$'th and $(j+1)$'th edges of $Q'$ (it is possible that $v_i=v_{i+1}$ when $i$ is even). 
Let $C_1,\ldots,C_5$ denote the five connected components obtained from $T$ after removing these four edges, where $v_1\in C_1$, $v_2,v_3 \in C_2$, $v_4,v_5\in C_3$, $v_6,v_7\in C_4$, $v_8\in C_5$. See Figure~\ref{fig}. 
  Let $C_{ij}$ denote the total weight of all edges of $G$ between $C_i$ and $C_j$. Notice that\\
\begin{center}$\begin{aligned}
M[i,j]&= C_{12} + C_{13} + C_{24}+ C_{25} + C_{34}+ C_{35},\\ 
M[i,j+1] &= C_{12} + C_{13} + C_{14}+ C_{25} + C_{35}+ C_{45}, \\ 
M[i+1,j] &= C_{13} + C_{23} + C_{34} + C_{35},\\ 
M[i+1,j+1] &= C_{13} + C_{14} + C_{23}+ C_{24}+C_{35} + C_{45},\\
\end{aligned}$\end{center}
and since $C_{24}\ge 0$ we get that $M[i,j] - M[i,j+1] \ge M[i+1,j] - M[i+1,j+1]$.
\end{proof}

\section{Unweighted Graphs and Dense Graphs}\label{sec:Improvements}

The main advantage of the approach of Mukhopadhyay and Nanongkai~\cite{Danupon} 
is that for restricted graph families they can plug in range counting/reporting structures with faster construction time. 

\subsection{Unweighted graphs}
For unweighted graphs (with parallel edges),~\cite{Danupon} used a two dimensional orthogonal range counting structure with faster preprocessing~\cite{CP10} and a data structure of~\cite{BGKS15} to devise a two dimensional orthogonal range sampling/reporting data structure with faster preprocessing. They plugged these improved data structures into their  algorithm for 2-respecting min cut, to obtain a running time of $O(m\sqrt{\log n}+n\log^{4}n)$ (multiply this by another $\log n$ factor for the running time of the resulting min cut algorithm). We show that an analogous  speedup can
be applied to our simplification (leading to an $O(m\sqrt{\log n}+n\log^{2}n)$ time algorithm). In fact, we only need the following range counting structure~\cite{CP10}:

\begin{lemma}[\cite{CP10}]
\label{lemma:chanp}
Given $m$ points in the 2D plane, we can construct in $O(m\sqrt{\log m})$ time a range counting structure
with $O(\log m/\log\log m)$ query time.
\end{lemma}

We use Lemma~\ref{lemma:chanp} instead of Chazelle's structure in the proof of Lemma~\ref{lemma:Chazelle}.
This decreases the overall running time to $O(m\sqrt{\log m}+n\log^{2}n+n\log n\log m/\log\log m)$. If $m\leq n^{2}\log n$, this is
$O(m\sqrt{\log n}+n\log^{2}n)$. Otherwise, we  replace the unweighted graph $G$ by a new weighted graph $G'$
with only $n^{2}$ edges (by collapsing parallel unweighted edges into a single weighted edge), and run the
previous algorithm in $O(n^{2}\log n+n\log^{2}n)=O(m+n\log^{2}n)$ time. This gives us
an $O(m\sqrt{\log n}+n\log^{2}n)$ time deterministic algorithm for 2-respecting min cut, and an $O(m\log^{3/2} n+n\log^{3}n)$ time randomized algorithm for min cut for unweighted undirected graphs.

\subsection{Dense weighted graphs}

We now present another speedup that can be applied to dense (weighted) graphs with $m = \Omega( n^{1+\epsilon})$, for
any $\epsilon>0$. For such graphs, we obtain  an  
$O(m\log n+n^{1+\epsilon})$-time algorithm for min cut.
We need the following structure:

\begin{lemma}
\label{lemma:fanout}
For any $\epsilon>0$, given $m\geq n$ weighted points in the $[n]\times [n]$ grid, we can construct in $O(m)$
time a data structure that reports the total weight of all points in any given rectangle $[x_{1},x_{2}]\times [y_{1},y_{2}]$
in $O(n^{\epsilon})$ time.
\end{lemma}

\begin{proof}
It is enough to construct a structure capable of reporting the total weight of all points in $[x,n]\times [y,n]$.
We use the standard approach of decomposing a 2D query into a number of 1D queries.

We start by designing a 1D structure storing a set $S$ of weighted numbers (weighted points in 1D) from $[n]$ that can be constructed in $O(|S|)$
time and returns the total weight of all numbers in $[x,n]$ in $O(n^{\epsilon})$ time.
Consider a complete tree $T$ of degree $B=n^{\epsilon}$ over the set of leaves $[n]$. Note that the depth of $T$ is $O(1/\epsilon)$. We construct and
store the subtree $T_{S}$ of $T$ induced by the leaves that belong to $S$. This takes $O(|S|)$ time and space and, assuming
that $S$ is given sorted, we can assume that the children of each node of $T_{S}$ are sorted.
Each node of $T_{S}$ stores the total weight of all numbers corresponding to its leaves.
Then, to find the total weight of numbers in $[x,n]$, we traverse $T_{S}$ starting from the root. 
Let $u$ be the current node of $T_{S}$. We scan the (at most) $B$ children of $u$ from right to left and,
as long as all leaves in their subtrees correspond to numbers from $[x,n]$, we add their stored
total weight to the answer. Then, if the interval of the leaves corresponding to the next child
intersects $[x,n]$ (but not entirely contained in $[x,n]$), we recurse on that child. Overall, there are at most $1/\epsilon$ steps, each taking
$O(B)$ time, so $O(n^{\epsilon})$ overall.

Our 2D structure for a set of $m$ weighted points from $[n]\times [n]$ uses the same idea.
We consider a complete tree $T$ of degree $B=n^{\epsilon}$ on the $y$ coordinates. We construct
and store the subtree $T'$ of $T$ induced by the $y$ coordinates of the points. At each
node $v$ of $T'$ we store a 1D structure responsible for all the points whose $y$ coordinate corresponds
to a leaf in the subtree of $v$. Overall, each point is stored at $1/\epsilon$ nodes of $T'$.
By first sorting all the points in $O(n+m)=O(m)$ time with radix sort we can assume that the points received
by every 1D structure are sorted, and construct all the 1D structures in $O(1/\epsilon\cdot m)=O(m)$ total time. 

Then, a query
 $[x,n]\times [y,n]$ is decomposed into $1/\epsilon\cdot B = O(n^{\epsilon})$
queries to the 1D structures by proceeding as above: we descend from the root, scanning the children
of the current node from right to left, issuing a 1D query to every child corresponding to
a $y$ interval completely contained in $[y,n]$, and then possibly recursing on the next child if its
$y$ interval intersects $[y,n]$. Each 1D query takes $O(n^{\epsilon})$ time, and there are $O(1/\epsilon\cdot n^{\epsilon})$
queries, so overall the query time is $O(n^{2\epsilon})$. By adjusting $\epsilon$ we obtain the lemma.
\end{proof}

By replacing Chazelle's structure with Lemma~\ref{lemma:fanout}, we obtain an algorithm with running time
$O(m+n\log^{2} n+n^{1+\epsilon}\log n)$. Because $m\leq n^{2}$, 
by adjusting $\epsilon$, this implies an $O(m\log n+n^{1+\epsilon})$-time algorithm for min cut for any constant $\epsilon>0$.

\bibliographystyle{plainurl}

\end{document}